%% file: UDG_journal_final.tex
\documentclass[runningheads]{llncs}
\usepackage[utf8]{inputenc}
\input{Macros}
\begin{document}
	\author{Jesper Nederlof\inst{1}\orcidID{0000-0003-1848-0076} \and Krisztina Szil\'agyi \inst{1}\orcidID{0000-0003-3570-0528}}
	\institute{Utrecht University, the Netherlands \email{\{j.nederlof, k.szilagyi\}@uu.nl}}
	
	\titlerunning{Detecting and Counting Small Patterns in Unit Disk Graphs}
	\authorrunning{J. Nederlof and K. Szil\'agyi}
	
	
	
	\title{Algorithms and Turing Kernels for Detecting and Counting Small Patterns in Unit Disk Graphs \thanks{Supported by the project CRACKNP that has received funding from the European Research Council (ERC) under the European Union’s Horizon 2020 research and innovation programme (grant agreement No 853234).}} 
	
	
	\maketitle
	\begin{abstract}
		In this paper we investigate the parameterized complexity of the task of counting and detecting occurrences of small patterns in unit disk graphs: Given an $n$-vertex unit disk graph $G$ with an embedding of ply $p$ (that is, the graph is represented as intersection graph with closed disks with unit diameter, and each point is contained in at most $p$ disks) and a $k$-vertex unit disk graph $P$, count the number of (induced) copies of $P$ in $G$.
		
		For general patterns $P$, we give an $2^{O(p  k /\log k)}n^{O(1)}$ time algorithm for counting pattern occurrences. We show this is tight, even for ply $p=2$: any $2^{o(n/\log n)}n^{O(1)}$ time algorithm violates the Exponential Time Hypothesis (ETH).
		
		For most natural classes of patterns, such as connected graphs and independent sets we present the following results: First, we give an $(pk)^{O(\sqrt{pk})}n^{O(1)}$ time algorithm, which is nearly tight under the ETH for bounded ply and many pattern classes. Second, for $p= k^{O(1)}$ we provide a Turing kernelization (i.e. we give a polynomial time preprocessing algorithm to reduce the instance size to $k^{O(1)}$).
		
		Our approach combines previous tools developed for planar subgraph isomorphism such as `efficient inclusion-exclusion' from [Nederlof STOC'20], and `isomorphisms checks' from [Bodlaender et al. ICALP'16] with a different separator hierarchy and a new bound on the number of non-isomorphic separations of small order tailored for unit disk graphs.
	\end{abstract}

\section{Introduction}
A well-studied theme within the complexity of computational problems on graphs is how much structure within inputs allows faster algorithms.
One of the most active research directions herein is to assume that input graphs are \emph{geometrically} structured. The (arguably) two most natural and commonly studied variants of this are to assume that the graph can be drawn on $\mathbb{R}^2$ without crossings (i.e., it is planar) or it is the intersection graph of simple geometric objects. While this last assumption can amount to a variety of different models, a canonical and most simple model is that of \emph{unit disk graphs}: Each vertex of the graph is represented by a disk with unit diameter and two vertices are adjacent if and only if the two associated disks intersect.

The computational complexity of problems on planar graphs has been a very fruitful subject of study: It led to the development of powerful tools such as \emph{Bakers layering technique} and \emph{bidimensionality} that gave rise to efficient approximation schemes and fast (parameterized) sub-exponential time algorithms for many NP-complete problems. One interesting example of such an NP-complete problem is \emph{(induced) subgraph isomorphism}: Given a $k$-vertex pattern $P$ and an $n$-vertex host graph $G$, detect or count the number of (induced) copies of $P$ inside $G$, denoted with $\sub(P,G)$ (respectively, $\ind(P,G)$). Here we think of $k$ as being much smaller than $n$, and therefore it is very interesting to obtain running times that are only exponential in $k$ (i.e. \emph{Fixed Parameter Tractable time}). This problem is especially appealing since it generalizes many natural NP-complete problems (such as Independent Set, Longest Path and Hamiltonian Cycle) in a natural way, but its generality poses significant challenges for the bidimensionality theory: It does not give sub-exponential time algorithms for this problem.

Only recently, it was shown in a combination of papers \cite{ICALP,DBLP:journals/siamcomp/FominLMPPS22, DBLP:conf/stoc/Nederlof20a} that, on planar graphs, subgraph isomorphism can be solved in $2^{\tilde{O}(\sqrt{k})}n^{O(1)}$ time for many natural pattern classes, and in $2^{O(k / \log k)}$ time for general patterns, complementing the lower bound of $2^{o(n / \log n)}$ time from~\cite{ICALP} based on the Exponential Time Hypothesis (ETH). It was shown in~\cite{DBLP:conf/stoc/Nederlof20a} that (induced) pattern occurrences can even be \emph{counted} in sub-exponential parameterized (i.e. $2^{o(k)}n^{O(1)}$) time.

Unfortunately, most of these methods do not immediately work for (induced) subgraph isomorphism on geometric intersection graphs: Even unit disk graphs of bounded ply\footnote{The ply of an embedded unit disk graph is defined as the the maximum number of times any point of the plane is contained in a disk.} are not $H$-minor free for any graph $H$ (which significantly undermines the bidimensionality theory approach), and unit disk graphs of unbounded ply can even have arbitrary large cliques. This hardness is inherent to the graph class: Independent Set is $W[1]$-hard on unit disk graphs and, unless ETH fails, it has no $f(k)n^{o(\sqrt{k})}$-time algorithm for any computable function $f$~\cite[Theorem 14.34]{CyganFKLMPPS15}.
On the positive side, it was shown in \cite{nevsetvril2012sparsity} that for bounded expansion graphs and fixed patterns, the subgraph isomorphism problem can be solved in linear time, which implies that subgraph isomorphism is fixed parameter tractable on unit disk graphs; however, their method relies on Courcelle's theorem and hence the dependence of $k$ in their running time is very large and far from optimal.

Therefore, a popular research topic has been to design such fast (parameterized) sub-exponential time algorithms for specific problems such as Independent Set, Hamiltonian Cycle and Steiner Tree~\cite{DBLP:journals/algorithmica/BhoreCKZ23,  DBLP:journals/siamcomp/BergBKMZ20,
DBLP:conf/soda/LokshtanovPSXZ22, DBLP:journals/jocg/ZehaviFLP021}.

In this paper we continue this research line by studying the complexity of the decision and counting version of (induced) subgraph isomorphism. While some general methods such as contraction decompositions~\cite{DBLP:conf/soda/Panolan0Z19} and pattern covering~\cite{DBLP:conf/esa/MarxP17} were already designed for graph classes that include (bounded ply) unit disk graphs, the fine-grained complexity of the subgraph problem itself restricted to unit disk graphs has not been studied and is still far from being understood.

\paragraph*{Our Results.}
To facilitate the formal statements of our results, we need the following definitions:
Given two graphs $P$ and $G$, we define 
\[
\begin{aligned}
\ind(P, G) &=\{f:V(P)\rightarrow V(G): f\text{ is injective}, uv \in E(P) \Leftrightarrow f(u)f(v)\in E(G)\},\\
\sub(P, G)&=\{f:V(P)\rightarrow V(G): f\text{ is injective}, uv \in E(P)\Rightarrow f(u)f(v)\in E(G)\}.
\end{aligned}
\]

We call elements of $\ind(P, G)$ and $\sub(P, G)$ \emph{pattern occurrences}.
Our main theorem reads as follows:

\begin{theorem}\label{thm:main}
	There is an algorithm that takes as input unit disk graphs $P$ and $G$ on $k$ vertices and $n$ vertices respectively, together with disk embeddings of ply $p$. It outputs $|\sub(P,G)|$ and $|\ind(P,G)|$ in $(pk)^{O(\sqrt{pk})}\sigma_{O(\sqrt{pk})}(P)^2n^{O(1)}$ time.
\end{theorem}

Note that this theorem can also be used to compute the number of (induced) subgraphs of $G$ isomorphic to $P$ by dividing $|\sub(P,G)|$ with $|\sub(P,P)|$ (and similarly dividing $|\ind(P,G)|$ with $|\ind(P,P)|$).
In this theorem, the parameter $\sigma_{s}$ is a somewhat technical parameter of the pattern graph $P$ that is defined as follows:

\begin{definition}[\cite{DBLP:conf/stoc/Nederlof20a}]
\label{def:sep}
    Given a graph $P$, we say that $(A, B)$ is a \emph{separation} of $P$ of order $s$ if $A\cup B=V(P)$, $|A\cap B|=s$ and there are no edges between $A\setminus B$ and $B\setminus A$. We say that $(A,B)$ and $(C,D)$ are \emph{isomorphic separations} of $P$ if there is a bijection $f:V(P)\leftrightarrow V(P)$ such that 
    \begin{itemize}
        \item For any $u,v\in V(P)$, $uv\in E(P)$ if and only if $f(u)f(v)\in E(P)$,
        \item $f(A)=C$, $f(B)=D$, 
        \item For any $u\in A\cap B$, $f(u)=u$.
    \end{itemize}
    We denote by $\mathcal{S}_s(P)$ a maximal set of pairwise non-isomorphic separations of $P$ of order at most $s$.  
    We define $\sigma_s(P)=|\mathcal{S}_s(P)|$ as the number of non-isomorphic separations of $P$ of order at most $s$.
\end{definition}

Note that $\sigma_s(P) \geq \binom{k}{s}s^s$.
For ply $p=O(1)$ and many natural classes of patterns such as independent sets, cycles or grids, it is easy to see that $\sigma_{O(\sqrt{k})}(P)$ is at most $2^{\tilde{O}(\sqrt{k})}$ and therefore Theorem~\ref{thm:main} gives a $2^{\tilde{O}(\sqrt{k})}n^{O(1)}$ time algorithm for computing $|\sub(P,G)|$ and $|\ind(P,G)|$. We also give the following new non-trivial bounds on $\sigma_s(P)$ whenever $P$ is a general (connected) unit disk graph:
\begin{theorem}\label{thm:isosep}
    Let $P$ be a $k$-vertex unit disk graph with an embedding of ply $p$, and $s$ be an integer. Then:
    (a) $\sigma_{s}(P)$ is at most $2^{O(s \log k +pk/\log k)}$, and (b) If $P$ is connected, then $\sigma_{s}(P)\leq 2^{O(s \log k)}$.
\end{theorem}
Using Theorem~\ref{thm:main}, this allows us to conclude the following result.

\begin{corollary}
There is an algorithm that takes as input a $k$-vertex unit disk graph $P$ and an $n$-vertex unit disk graph $G$, together with their unit disk embeddings of ply $p$, and outputs $|\sub(P, G)|$ and $|\ind(P, G)|$ in $2^{O(pk/ \log k)}n^{O(1)}$ time.  
\end{corollary}

We show that this cannot be significantly improved even if $P$ and $G$ have ply two:
\begin{theorem}\label{thm:LB}
	Assuming ETH, there is no algorithm to determine if $\sub(P,G)$ ($\ind(P,G)$ respectively) is nonempty for given $n$-vertex unit disk graph $G$ and a unit disk graph $P$ in $2^{o(n/\log n)}$ time, even when $P$ and $G$ have a given embedding of ply $2$.
\end{theorem}
Note that the ply of $G$ is 1 if and only if $G$ is an independent set so the assumption on the ply in the above statement is necessary. To our knowledge, this is the first lower bound based on the ETH excluding $2^{O(\sqrt{n})}$ time algorithms for problems on unit disks graphs (of bounded ply), in contrast to previous bounds that only exclude $2^{o(\sqrt{n})}$ time algorithms. 

Clearly, a unit disk graph $G$ of ply $p$ has a clique of size $p$, i.e. its clique number $\omega(G)$ is at least $p$. On the other hand, it was shown in \cite{har2021stabbing} that $p\geq \omega(G)/5$. In other words, parameterizations by ply and clique number are equivalent up to a constant factor.

\subsection{Our techniques} Our approach heavily builds on the previous works \cite{ICALP, DBLP:journals/siamcomp/FominLMPPS22, DBLP:conf/stoc/Nederlof20a}:
Theorem~\ref{thm:main} is proved via a combination of the dynamic programming technique from~\cite{ICALP} that stores representatives of non-isomorphic separations to get the runtime dependence down from $2^{O(k)}$ to $\sigma_{O(\sqrt{pk})}(P)$, and the \emph{efficient inclusion-exclusion} technique from~\cite{DBLP:conf/stoc/Nederlof20a} to solve counting problems (on top of decision problems).\footnote{Similar to what was discussed in~\cite{DBLP:conf/stoc/Nederlof20a}, this technique seems to be needed even for simple special cases of Theorem~\ref{thm:main} such as counting independent sets on subgraphs of (subdivided) grids.} We combine these techniques with a divide and conquer strategy that divides the unit disk graph in smaller graphs using horizontal and vertical lines as separators.
As a first step in our proof, we give a \emph{Turing kernelization} for the counting versions of (induced) subgraph isomorphism that uses efficient inclusion-exclusion. 
Theorem~\ref{thm:isosep} uses a proof strategy from~\cite{ICALP} combined with a bound from~\cite{number_graphs} on the number of non-isomorphic unit disk graphs.
Theorem~\ref{thm:LB} builds on a reduction from~\cite{ICALP}, although several alterations are needed to ensure the graph is a unit disk graph of bounded ply.

\paragraph*{Organization.} 
In Section~\ref{sec:prel} we provide additional notation and some preliminary lemmas.  
In Section~\ref{sec:Turing} we provide a Turing kernel.
In Section~\ref{sec:main} we build on Section~\ref{sec:Turing} to provide the proof of Theorem~\ref{thm:main}. In Section~\ref{sec:boundsigma} we give a bound on the parameter $\sigma_s$.
In Section~\ref{sec:LB} we give a lower bound, i.e. prove Theorem~\ref{thm:LB}.
We finish with some concluding remarks in Section~\ref{sec:concUDG}. 

\section{Preliminaries}\label{sec:prel}
\textbf{Notation.} Given a graph $G$ and a subset $A$ of its vertices, we define $G[A]$ as the subgraph of $G$ induced by $A$. 

Given a unit disk graph $G$, we say that $G$ has \emph{ply} $p$ if there is an embedding of $G$ such that every point in the plane is contained in at most $p$ disks of $G$. Given an embedding of a unit disk graph $G$ and a vertex $v\in V(G)$, we define $D(v)$ as the unit disk corresponding to $v$. 
Let $P$ and $G$ be unit disk graphs and let $|V(G)|=n$, $|V(P)|=k$. 

 We denote all vectors by bold letters, the all ones vector by $\mathbf{1}$ and the all zeros vector by $\mathbf{0}$.
We use the Iverson bracket notation: for a statement $P$, we define $[P]=1$ if $P$ is true and $[P]=0$ if $P$ is false. We define $[k]=\{1,\dots,k\}$. Given a function $f:A\rightarrow B$ and $C\subseteq A$, we define the restriction of $f$ to $C$ as $f|_C:C\rightarrow B$, $f|_C(c)=f(c)$ for all $c\in C$. If $g = f|_C$ for some $C$ then we say that $f$ \emph{extends} $g$.

 We use standard notation for path decompositions (as for example outlined in~\cite[Section 7.2]{CyganFKLMPPS15}). For completeness, recall the following definition:
\begin{definition}[Path Decompositions]
	A \emph{path decomposition} of a graph $G$ is a sequence $\mathcal{P}=(X_1,\ldots,X_r)$ of \emph{bags}, where $X_i \subseteq V(G)$ for each $i \in \{1,2,\ldots,r\}$ such that $\cup_{i=1}^r X_i = V(G)$, for every $uv \in E(G)$ there exists a bag containing both $u,v$ and for every $u \in V(G)$ the set of bags of $u$ induce a path in $\mathcal{P}$.
\end{definition}

\begin{definition}
    For integers $b,h$, we say a unit disk graph $G$ of ply $p$ can be \emph{drawn in a $(b \times h)$-box} with ply $p$ if it has an embedding of ply $p$ as unit disk graph in a $b\times h$ rectangle. 
\end{definition}

Throughout this paper, we assume that the sides of the box are axis parallel, and the lower left corner is at $(0,0)$. We assume that if a graph $G$ can be drawn in a $(b\times h)$-box then we are given such an embedding.

\begin{lemma}\label{lem:pw}
Given a unit disk graph $G$ with a drawing in a $(b\times h)$-box with ply $p$, one can construct in polynomial time a path decomposition of $G$ of width $4(\min\{b,h\}+1)p$.
\end{lemma}
\begin{proof}
	 Without loss of generality, we may assume $h\leq b$.  
	 For $i=1,\dots b$ define 
	$$L_i=\{(x,y)\in\mathbb{R}^2:\: i-1\leq x\leq i \},$$
	$$X_i=\{v\in V(G): \: D(v)\cap L_i\neq \emptyset \}.$$
	It is easy to see that $X_1,\dots, X_b$ is a path decomposition of $G$. Let us now bound the size of $X_i$. Let $S_i=\{(x,y)\in \mathbb{Z}^2:\: x\in\{i-2, i-1, i, i+1 \}  \text{ and } 0\leq y\leq h\}$. Each disk in $X_i$ contains a point from $S_i$, so $|X_i|\leq p\cdot |S_i|=4p(h+1)$. 
\end{proof}

\begin{lemma}[Theorem 6.1 from~\cite{number_graphs}]\label{lem:number_udg} Let a non-decreasing bound $b=b(n)$ be given, and let $\mathcal{U}_n$ denote the set of unlabeled unit disk graphs on $n$ vertices with maximum clique size at most $b$. Then
    $|\mathcal{U}_n|\leq 2^{12(b+1)n}.$
\end{lemma}

\textbf{Non-isomorphic Separations.} In this paper we will work with non-isomorphic separations of small order, as defined in Definition~\ref{def:sep}. 

\begin{observation} \label{obs:iso}
Given a $k$-vertex graph $P$ and separations $(A, B)$ and $(C, D)$ of $P$, one can check in quasi-polynomial time in $k$ if $(A, B)$ and $(C, D)$ are isomorphic separations and, if yes, find the corresponding bijection.
\end{observation}
The above observation can be shown using the quasi-polynomial time algorithm for graph isomorphism from~\cite{DBLP:conf/stoc/Babai16}, extended to the colored subgraph isomorphism problem by a standard reduction (see~\cite[Theorem 1]{schweitzer2009problems}).

For separations $(A, B), (A', B')\in \mathcal{S}_{s}(P)$, we define 
$$\mu((A,B), (A',B'))=|\{(C,D):\: (C,D) \text{ is a separation of }P\text{ such that }$$
$$C\subseteq A'\text{ and } (C,D) \text{ isomorphic to } (A,B)\}|.$$

\begin{lemma}\label{lem:iso}
    Given a graph $P$, one can compute $\mathcal{S}_s(P)$ and for each pair of separations $(A,B),(A',B') \in \mathcal{S}_s(P)$ the multiplicity $\mu((A,B),(A',B'))$ in time $\sigma_s(P) n^{O(1)}$.
\end{lemma}
\begin{proof}

The proof is identical the proof of~\cite{DBLP:conf/stoc/Nederlof20a}, but recorded here for completeness.
We start enumerating separations $(C,D)$ by iterating over all candidates $X$ for $C \cap D$. There are $\binom{k}{s}$ possibilities since we are interested in separations of order $s$. Now any separation $(C,D)$ of $P$ with $C \cap D = X$ is formed by selecting a subset of the connected components of $P - X$ and adding it to $C$ (the remaining connected components will be added to $D$). 
Given two connected components $P_1$ and $P_2$ of $P-X$, we can check in quasi-polynomial time in $k$ whether there is an isomorphism $f:V(P)\rightarrow V(P)$ such that $f|_X$ is the identity function and $f(P_1)=f(P_2)$: Indeed, we apply Observation~\ref{obs:iso} to $(P_1\cup X, (P-P_1)\cup X)$ and $(P_2\cup X, (P-P_2)\cup X)$. This isomorphism relation is an equivalence relation, and we define for each equivalence class a unique representative denoted with $\mathcal{P}_i$.

Let $\mathcal{P}_1,\dots, \mathcal{P}_{t}$ be the computed set of representatives of all isomorphism classes of connected components of $P-X$, and let $p_i$ be the number of connected components of $P - X$ isomorphic to $\mathcal{P}_i$. 

Now we can construct $\mathcal{S}_s(P)$ as follows:
For each vector $t$-dimensional vector $(v_1,\dots, v_t)\in \{0,\dots, p_1\}\times \cdots\times\{0,\dots, p_t\}$, we add to $\mathcal{S}_s(P)$ the separation $(C,D)$ where $C$ is obtained by adding for each $i=1,\ldots,t$ exactly $v_i$ copies of a connected component isomorphic to $\mathcal{P}_i$ to $C$ (and the remaining $p_i-v_i$ components to $D$). This concludes the construction of $\mathcal{S}_s$.

The multiplicity $\mu((A,B),(A',B'))$ can be computed in a similar way: For each candidate $X$ for $A \cap B \subseteq A'$ we compute for each two connected components of $P - X$ that do not contain elements from $B'$ whether they are isomorphic (analogous to how it was done above).
Let $\mathcal{P}_1,\dots, \mathcal{P}_{t}$ be the computed set of representatives of all isomorphism classes of connected components of $P-X$ that do not intersect $B'$, and let $p_i$ be the number of connected components of $P - X$ not intersecting $B'$ that are isomorphic to $\mathcal{P}_i$. If $A$ contains $v_i$ components that are isomorphic to $\mathcal{P}_i$, then we can compute 
\[
\mu((A,B),(A',B')) = \prod_{i=1}^t \binom{p_i}{v_i}.
\]
It is easy to see that all the above can be done in the claimed time bound (since $\sigma_s(P) \geq \binom{k}{s}s^s$).
\end{proof}

\textbf{Solving instances where $G$ has low pathwidth.}
The following lemma can be shown with standard dynamic programming over tree decompositions:
\begin{lemma}
\label{lem:smallalg}
Given $P,G$, and $I\subseteq V(G)$, a path decomposition of $G$ of width $t$, we can compute  $|\{g\in \sub(P, G): g \text{ extends } f \}|$ and $|\{g\in \ind(P, G): g \text{ extends } f \}|$ for every $P'\subseteq V(P)$ of size at most $|I|$ and injective function $f:P'\rightarrow I$ in time $\sigma_{t}(P)(t+1)^t(|V(P)|+1)^{|I|}n^{O(1)}$.
\end{lemma}
\begin{proof}
	We present only the proof for the subgraph case, as the proof for induced case is completely analogous.
    It is well known (see for example~\cite[Lemma 7.2]{CyganFKLMPPS15}) that each path decomposition can be modified to a path decomposition of the same width such that for each bag $i$ either we have $X_i = X_{i-1} \cup \{v\}$ or $X_i = X_{i-1} \setminus \{v\}$. We say $X_i$ is an \emph{introduce} bag in the first case, and a \emph{forget bag} in the second case.
	
	For a bag $i$, we define $G_i = G[\cup_{j=1}^i X_j]$. For each bag $i$, separation $(A,B)$ of $P$ of order $t$, and functions $f:P' \rightarrow I$ (where $P' \subseteq V(P)$) and $h: A \cap B \rightarrow X_i$ we define
	\[
	T_i[(A,B),f ,h] = |\{ g \in \sub(P[A],G_i) : g|_{A\cap P'}=f|_{A\cap P'} \text{ and } g \text{ extends } h \}|.
	\]
	\begin{description}
		\item[Leaf.]
		This corresponds to $i=1$ and $G_1$ being an empty graph. Then we have $T_1[(A,B),f,h]=1$ for $A=\emptyset$ and $f,h$ being the empty function, and it is equal to $0$ otherwise.
		\item[Introduce Vertex.]
		This corresponds to $G_{i}=G_{i-1} \cup v$ for some vertex $v$.
		If $h^{-1}(v)=\emptyset$ we have
		\[
		T_i[(A,B),f,h] = T_{i-1}[(A,B),f,h].
		\]
		Otherwise, if $h$ is not an injective function, or there is a vertex $u \in A \cap N(h^{-1}(v))$ such that $\{h(u),v\} \not\in E(G)$, or $f^{-1}(v)$ is non-empty but not equal to $h^{-1}(v)$, then $T_i[(A,B),f,h]=0$ since $h$ cannot be extended to any function counted in $T_i[(A,B),f,h]$.
		
		Otherwise, let $h^{-1}(v)=\{a\}$. We have that
		\[
		T_i[(A,B),f,h]= T_{i-1}[(A \setminus \{a\},B \cup \{a\} ),f,h|_{(A\cap B) \setminus \{a\}}] 
		\]
		since any function $g$ counted in $T_i[(A,B),f,h]$ needs to map a vertex in $A$ to $v$ and the injectivity and adjacencies and extension properties involving $v$ are satisfied by assumption.

		\item[Forget Vertex.]
		This corresponds to $X_{i}=X_{i-1} \setminus \{v\}$ for some vertex $v$. In this case we need to decide the preimage of $v$ (if any):
		\[
		T_i[(A,B),f,h] = T_{i-1}[(A,B),f,h] + \sum_{u \in A \setminus B} T_{i-1}[(A, B \cup \{u\}),f,h'],
		\]
		where $h'$ is obtained from $h$ by extending its domain with $u$ and defining $h'(u)=v$. If $u$ is the preimage of $v$ is a function $g$ contributing to $T_i[(A,B),f,h]$, then $u \in A$ since it is mapped to a vertex of $G_i$, and it is not in $A \cap B$ because $h$ then already maps $u$ to a vertex in $X_i$ (and since $X_i$ does not contain $v$, it is therefore it is not well-defined).
	\end{description}

Note that currently the first index of the table entries $T_i$ goes over all separations of $P$ of order $t$, which could be too large for our purposes. We overcome this issue by showing that it suffices to compute the table entries indexed by separations in $\mathcal{S}_t(P)$.
Formally, given a separation $(A, B)$ of $P$ of order at most $t$ and functions $f:P'\rightarrow I$, $h:A\cap B\rightarrow X_i$, we claim that we can find in quasi-polynomial time a separation $(C, D)\in \mathcal{S}_t(P)$ and $f':P''\rightarrow I$ such that $T_i[(A, B), f, h]=T_i[(C, D), f', h]$. 
Indeed, using Observation~\ref{obs:iso}, we can check for each separation in $\mathcal{S}_t(P)$ whether it is isomorphic to $(A, B)$. Once we found a separation $(C, D)\in\mathcal{S}_t(P)$ that is isomorphic to $(A, B)$ and the corresponding isomorphism $\phi:V(P)\rightarrow V(P)$ that maps $(A, B)$ to $(C, D)$, we define $f'$ as the composition of $\phi$ and $f$ and $P''=\phi(P')$. Note that by construction we have $C \cap D=A\cap B$ and $\phi|_{A\cap B}$ is the identity function.     
 

Now the desired value $|\{g\in \sub(P,G)\: : \: g \text{ extends } f\}|$ can be found as the table entry $T_i[(V(P),\emptyset),f, h]$, where $i$ is the last bag of the path decomposition and $h$ the empty function.
	The number of table entries is at most $\sigma_t(P) (t+1)^t (|V(P)|+1)^{|I|} n$, and therefore the running time of the algorithm is at most $\sigma_t(P) (t+1)^t (|V(P)|+1)^{|I|} n^{O(1)}$.
\end{proof}

The following lemma simply states that a long product of matrices can be evaluated quickly, but is nevertheless useful as subroutine in the `efficient inclusion-exclusion' technique.

\begin{lemma}[\cite{DBLP:conf/stoc/Nederlof20a}] \label{lem:inc_exc}
Given a set $A$, an integer $h$ and a value $T[x,x'] \in \mathbb{Z}$ for every $x,x'\in A$, the value
\begin{equation}\label{eq:prod}
\sum_{x_1,\ldots,x_h \in A} \prod_{i=1}^{h-1} T[x_i,x_{i+1}]
\end{equation}
can be computed in $O(h|A|^2)$ time.
\end{lemma}

\section{Turing kernel}\label{sec:Turing}


We will now present a preprocessing algorithm for computing $|\sub(P,G)|$ (the algorithm for $|\ind(P,G)|$ is analogous) that allows us to assume that $G$ can be drawn in a $(O(k)\times O(k))$-box with ply $p$, i.e. that $|V(G)|=O(k^2p)$. 
This can be seen as a polynomial Turing kernel in case $p$ and $\sigma_0(P)$ are polynomial in $k$ (note that $\sigma_0(P)$ could be exponentially large in $k$). A \emph{Turing kernel} of size $f(k)$ is an algorithm that solves the given problem in polynomial time, when given access to an oracle that solves instances of size at most $f(k)$ in a single step. 

Lemma \ref{lem:Tker_small} describes how to reduce the width of $G$ (and analogously the height of $G$). To prove it, we use the shifting technique. This general technique was first used by Baker \cite{baker1994approximation} for covering and packing problems on planar graphs and by Hochbaum and Maass \cite{hochbaum1985approximation} for geometric problems stemming from VLSI design and image processing. 

Intuitively, we draw the graph on a grid, and delete all the disks that intersect certain columns of the grid. After doing that, the remaining graph will consist of several small disconnected "building blocks". Each connected component of the pattern will be fully contained in one of the blocks, and since the blocks are small we can use the oracle to count the number of these occurrences. We take advantage of the fact that we can group together connected components that are isomorphic.
We use Lemma \ref{lem:Tker_small} twice, to reduce the width and height of $G$ to $O(k)$.
\begin{lemma}\label{lem:Tker_small}
    Suppose we have access to an oracle that computes $|\sub (P, G')|$ in constant time, where the host graph $G'$ can be drawn in a $(h\times O(k))$-box for some $h$ with ply $p$. Then we can compute $|\sub(P,G)|$ for host graphs $G$ of ply $p$ in time $n\cdot k^{O(1)}\cdot\sigma_0(P)^2$. 
\end{lemma}
\begin{proof}
    For $i\in\{0,\dots, 2k\}$, let 
    $C_i=\{(x,y)\in\mathbb{R}^2:\: x\equiv i\: (\mod\ 2k+1  ) \}.$
    Informally, we draw a grid and select every $(2k+1)$th vertical gridline. Let $P_i$ be the set of all pattern occurrences whose image is disjoint from $C_i$:    
    $$P_i=\{f:V(P)\rightarrow V(G):\: f\in \sub(P, G)\text{ and } \forall v\in \im(f): D(v)\cap C_i=\emptyset  \}.$$

    Note that every disk intersects at most two grid lines, so by the pigeonhole principle we have that $\sub(P, G)=\cup_{i=0}^{2k} P_i$.    
    By the inclusion-exclusion principle, $\left| \bigcup_{i=0}^{2k} P_i\right|$ equals
    \begin{equation}\label{eq:Tk_inc_exc}
      \bigsum_{\emptyset \subset C\subseteq \{0,\dots,2k\}}(-1)^{|C|} \left| \bigcap_{i\in C} P_i \right| =\bigsum_{\ell=1}^{2k+1}(-1)^{\ell} \bigsum_{0\leq c_1<\dots< c_{\ell}\leq 2k} \left| \bigcap_{j\in\{c_1,\dots,c_{\ell} \}} P_j\right|.
    \end{equation}

    Let us show how we can compute $|\cap_{j\in\{c_1,\dots,c_{\ell} \}} P_j|$ quickly. For $a,b\in \{0,...,2k\}$, we define $B[a,b]$ to be
    $$\begin{cases}
    \{(x,y)\in\mathbb{R}^2: \: (\exists t\in \mathbb{N}_0)\: a+(2k+1)t \hspace{2.5em}   <x<b+(2k+1)t \}, & \text{ if } a\leq b, \\
    \{(x,y)\in \mathbb{R}^2:\: (\exists t\in\mathbb{N}_0)\: a+(2k+1)(t-1)<x<b+(2k+1)t\}, & \text{ if } a>b.
    \end{cases}$$
    
    We define $\mathcal{B}[a,b]$ as the induced subgraph of $G$ such that all its disks are fully contained in $B[a,b]$, i.e. the subgraph induced by the vertex set $\{v\in V(G):\: D(v)\subseteq B[a,b]\}$.
    These sets are our "building blocks": after deleting $C_{c_1},\dots C_{c_{\ell}}$, the remaining graph is $\cup_{\alpha=1}^{\ell}\mathcal{B}[c_\alpha, c_{\alpha+1}]$, where we define $c_{\ell+1}=c_1$. 

    Let $t$ be the number of non-isomorphic connected components of $P$ and let $\mathcal{C}_0(P)=\{\mathcal{P}_1,\dots, \mathcal{P}_{t} \}$ be the set of representatives of all isomorphism classes of connected components of $P$. We can encode $P$ as vector $\mathbf{p}=(p_1,\dots, p_{t})$, where $p_i$ is the size of the isomorphism class of $\mathcal{P}_i$. 

    Let $U=\{0,\dots, p_1\}\times \cdots\times\{0,\dots, p_t\}$. 
    For a $t$-dimensional vector $(v_1,\dots, v_t)\in U$ we define $P[(v_1,\dots, v_{t})]$ as the subgraph of $P$ that contains $v_i$ copies of $\mathcal{P}_i$. 

    We would like to count in how many ways can we distribute the connected components of $P$ to the building blocks.  
    Equivalently, we can count the number of ways to assign a vector $\mathbf{v}^{\alpha}\in U$ to each block $\mathcal{B}[c_\alpha, c_{\alpha+1}]$ such that $\sum \mathbf{v}^{\alpha}=\mathbf{p}$. 
    
    Thus we have
    \begin{equation}\label{eq:Tk_bb}
        \left|\bigcap_{j\in \{c_1,\dots,c_{\ell}\}} P_j \right|
        =\bigsum_{\mathbf{v}^1+\dots+\mathbf{v}^{\ell}=\mathbf{p}} \quad\bigprod_{\alpha=1}^{\ell} |\sub(P[\mathbf{v}^\alpha], \mathcal{B}[c_{\alpha}, c_{\alpha+1}])|.
    \end{equation} 

    Note that $|U|=(p_1+1)\cdots (p_t+1)=\sigma_0(P)$: indeed, every vector $\mathbf{u}\in U$ corresponds to a unique separation $(V(P'), V(P-P'))$ of order 0, where $P'$ consists of $u_i$ copies of $\mathcal{P}_i$. 
    Combining~\eqref{eq:Tk_inc_exc} and~\eqref{eq:Tk_bb}, we get that 
    $$\left|\bigcup_{i=0}^{2k} P_i\right|=\sum_{\ell=1}^{2k} (-1)^{\ell} T_{\ell},$$
    where 
    \begin{equation}\label{eq:Tl}
    T_{\ell}=\bigsum_{\substack{0\leq c_1<\dots<c_{\ell}\leq 2k \\ \mathbf{v}^1+\dots+\mathbf{v}^{\ell}=\mathbf{p}} } \quad \bigprod_{\alpha=1}^\ell |\sub(P[\mathbf{v}^{\alpha}], \mathcal{B}[c_{\alpha}, c_{\alpha+1}])|.
    \end{equation}
    Suppose for now that we have computed $|\sub(P[\mathbf{v}^\alpha], \mathcal{B}[a,b])|$ for all $a,b\in \{0,\dots,2k\}$, $\mathbf{v}^\alpha\in U$, and that we want to compute $T_\ell$ quickly.  
    
    To apply Lemma~\ref{lem:inc_exc}, we have to rewrite the sum~\eqref{eq:Tl} in such a way that the variables are pairwise independent. We replace the condition $c_i<c_{i+1}$ by multiplying with $[c_i<c_{i+1}]$. To replace the condition on the variables $\mathbf{v}^i$, we will re-index these variables by $\mathbf{u}^1, \dots, \mathbf{u}^\ell$, where $\mathbf{u}^{i}=\sum_{j=1}^i \mathbf{v}^j$ for $i\in [\ell-1]$ and $\mathbf{u}^\ell=\mathbf{p}-\mathbf{u}^{\ell-1}$, $\mathbf{u}^0=\mathbf{0}$.  
    Therefore, we have 
    
    $$T_{\ell}=\bigsum_{\substack{c_1,\dots, c_\ell\in \{0,\dots,2k\}\\ \mathbf{u}^1,\dots,\mathbf{u}^{\ell-1}\in U}}\: |\sub(P[\mathbf{p}-\mathbf{u}^{\ell-1}], \mathcal{B}[c_\ell, c_1])|\: \bigprod_{i=1}^{\ell-1} [c_i<c_{i+1}]\cdot[\mathbf{u}^{i-1}\leq \mathbf{u} ^{i}]$$
    $$\cdot |\sub(P[\mathbf{u}^{i}-\mathbf{u}^{i-1}], \mathcal{B}[c_{i}, c_{i+1}]) |.$$
     
    By Lemma \ref{lem:inc_exc}, we can compute $T_{\ell}$ in time $\ell\cdot ((2k+1)\cdot \sigma_0(P))^2 $ if we are given $|\sub(P[\mathbf{u}], \mathcal{B}[a,b])|$ for all $\mathbf{u}\in U$, $a,b\in \{0,\dots,2k\}$. 

    It remains to show how we can compute $|\sub(P[\mathbf{u}], \mathcal{B}[a,b])|$ for given $\mathbf{u}\in U$, $a,b\in \{0,\dots,2k\}$. 
    Let $C_1,\dots, C_d$ be the connected components of $\mathcal{B}[a,b]$. Note that each $C_i$ can be drawn in a $(h\times O(k))$-box with ply $p$ for some $h$, so we can use the oracle to compute $|\sub(P[\mathbf{w}], C_i)|$ for all $\mathbf{w}\in U$, $i\in [d]$.
    We would like to distribute the connected components of $P[\mathbf{u}]$ to $C_1,\dots C_d$. We can do this by dynamic programming. For $i\in [d]$ and $\mathbf{w}\in U$, we define
    $$T'[i, \mathbf{w}]=|\sub(P[\mathbf{w}], C_1\cup\dots\cup C_i)|$$
    The recurrence is as follows:
    $$T'[i, \mathbf{w}]=\sum_{\mathbf{w'}\leq \mathbf{w}}|\sub(P[\mathbf{w'}], C_i)|\cdot T[i-1, \mathbf{w}-\mathbf{w'}],$$
    where $\textbf{w} \leq \textbf{w'}$ indicates that $\textbf{w}$ is in each coordinate smaller than $\textbf{w'}$. 
    Thus we can compute $T'[i, \mathbf{u}]$ in time $d|U|^2=d\sigma_0(P)^2\leq (n/2k)\sigma_0(P)^2$. 

    Therefore, we can compute $|\sub(P,G)|$ in time $n\cdot k^{O(1)}\cdot\sigma_0(P)^2$.
\end{proof}

\begin{theorem} \label{thm:turing} 
    For unit disk graphs $P$ and $G$ with given embedding of ply $p$, $|\sub(P, G)|$ can be computed in time $\sigma_0(P)^2 \cdot n \cdot k^{O(1)}$ when given access to an oracle that computes $|\sub(P, G)|$ where the host graph has size $O(k^2p)$ in constant time. In particular, there is a Turing kernel for computing $|\sub(P, G)|$ when $\sigma_0(P)$ and $p$ are polynomial in $k$.  
\end{theorem}
\begin{proof}
    Using Lemma~\ref{lem:Tker_small}, we reduce the problem to computing $|\sub(P, G)|$ for graphs $G$ that can be drawn in a $(h\times O(k))$-box with ply $p$. Applying Lemma~\ref{lem:Tker_small} one more time, we reduce the problem to computing $|\sub(P, G)|$ for graphs $G$ that can be drawn in a $(O(k)\times O(k))$-box with ply $p$, i.e. where $|V(G)|=O(k^2p)$. 

\end{proof}

\section{Proof of Theorem~\ref{thm:main}: The algorithm} 
\label{sec:main}


We present only the proof for $\sub(P, G)$, since the proof for $\ind(P, G)$ is analogous.
Before we start with the proof, we need to give a number of definitions: 
Suppose that a unit disk embedding of $G$ in a $(b\times h)$-box with ply $p$ is given. 

For integers $0\leq x\leq x'\leq b$, we define $V\langle x,x'\rangle\subseteq V(G)$ as the set of vertices $v\in V(G)$ such that $D(v)\cap \{(x_0,y_0) 
\in \mathbb{R}^2:\: x\leq x_0\leq x' \}\neq \emptyset$. Informally, $V\langle x,x'\rangle$ consists of vertices whose associated unit disks are (partially) between vertical lines $x$ and $x'$. We define $G\langle x,x'\rangle=G[V\langle x,x'\rangle]$ and $V\langle x\rangle=V\langle x,x\rangle$, $G\langle x\rangle=G\langle x,x\rangle$. 


Given functions $f_1:D_1\rightarrow R_1$ and $f_2:D_2\rightarrow R_2$, we say $f_1$ and $f_2$ are \emph{compatible} if 
\begin{itemize}
    \item for all $u\in D_1\cap D_2$, $f_1(u)=f_2(u)$, and
    \item for all $r\in R_1\cap R_2$, we have $f_1^{-1}(r)=f_2^{-1}(r)$.
\end{itemize}
If $f_1, f_2$ are compatible, we define $f=f_1\cup f_2$ as the function with domain $D_1 \cup D_2$ satisfying $f|_{D_1}=f_1$ and $f|_{D_2}=f_2$.

Note that in the above definition, the choice of $R_1$ and $R_2$ matter. For example, the identity functions $f_1:\{1\}\rightarrow \{1\}$ and $f_2:\{2\}\rightarrow \{2\}$ are compatible, but the identity functions $f_3:\{1\}\rightarrow \{1,2\}$ and $f_4:\{2\}\rightarrow \{1,2\}$ are not compatible, since $f_3^{-1}(1)=\{1\}$ but $f_4^{-1}(1)=\emptyset$.

Using Theorem \ref{thm:turing}, we can assume that $G$ can be drawn in a $(O(k)\times O(k))$-box with ply $p$. We will use dynamic programming. 
We will first define the sets of partial solutions that are counted in this dynamic programming algorithm.
For variables
\begin{itemize}
	\item integers $0\leq x < x'\leq b$, 
	\item separation $(A,B)$ of $P$ of order at most $2\sqrt{pk}$,
	\item injective $f:A\cap B\rightarrow V\langle x\rangle \cup V\langle x'\rangle$ such that $|f^{-1}(V\langle x\rangle)|, |f^{-1}(V\langle x'\rangle)|\leq \sqrt{pk} $
\end{itemize}
we define
\[
T[x, x', (A,B), f] =  \{g \in \sub(P[A],G\langle x,x'\rangle) : g\textnormal{ extends } f \}.
\]
Note that $T$ is indexed by \emph{any} separation of $P$ of order $2\sqrt{pk}$. We will later replace this with a set of non-isomorphic separations to obtain the claimed $\sigma_{2\sqrt{pk}}(P)$ dependence in the running time.

Informally, $T[x, x', (A,B), f]$ is the set of all  occurrences of $P[A]$ in $G\langle x, x'\rangle$ such that $f$ describes their behaviour on the "boundary" $V\langle x \rangle\cup V\langle x'\rangle$. 
We will now show how to compute the table entries. We consider two cases, depending on whether $x'-x$ is less than $\sqrt{k/p}$ or not. 

\textbf{Case 1}: $x'-x\leq \sqrt{k/p}$ \\
Note that in this case, the pathwidth $t$ of $G\langle x, x'\rangle$ is $O(\sqrt{pk})$ by Lemma~\ref{lem:pw}. For a fixed subset $I\subseteq V\langle x\rangle\cup V\langle x'\rangle$ of size $2\sqrt{pk}$, we apply Lemma~\ref{lem:smallalg} to compute $T[x,x', (A, B), f]=\{g\in\sub(P[A], G\langle x,x'\rangle):\: g\text{ extends }f \}$ for every injection $f:A\cap B\rightarrow I$ in time $\sigma_{t}(P[A])(t+1)^t(|A|+1)^{2\sqrt{pk}}|V\langle x,x'\rangle|^{O(1)}$. Applying Lemma~\ref{lem:smallalg} to every subset $I\subseteq V\langle x\rangle \cup V\langle x'\rangle$ of size $2\sqrt{pk}$, we can compute $T[x,x', (A, B), f]$ for every $f$ in time ${2pk\choose 2\sqrt{pk}}\sigma_{t}(P)(\sqrt{pk})^{O(\sqrt{pk})}|V\langle x,x'\rangle|^{O(1)}$. Using $t=O(\sqrt{pk})$, $|V\langle x,x'\rangle|=O(k\sqrt{kp})$ and a well-known bound on binomial coefficients, ${a\choose b}\leq (\frac{ae}{b})^b$, the above running time is bounded by $(pk)^{O(\sqrt{pk})}\sigma_{O(\sqrt{pk})}(P)$.


\textbf{Case 2}: $x'-x>\sqrt{k/p}$ \\
Let $g \in T[x,x',(A,B),f]$, and let $Q=\im(g)$. For $m\in\{x+1,\dots, x'-1\}$, we say that $Q$ is \emph{sparse at} $m$ if $|Q\cap V\langle m\rangle|\leq 2\sqrt{pk}$, i.e. the vertical line at $m$ intersects at most $2\sqrt{pk}$ disks in $Q$. Since $|Q|\leq k$, every disk intersects at most two grid lines and $x' - x > \sqrt{k/p}$, there is at least one $m$ such that $Q$ is sparse at $m$ by the averaging principle.  Therefore, 
\[
T[x,x', (A,B), f]=\bigcup_{m=x+1}^{x'-1}\{g\in T[x,x', (A,B), f]: \: g(A)\text{ is a sparse at } m \}.
\]
By the inclusion-exclusion principle, $|T[x, x', (A, B), f]|$ is equal to 
\[
\bigsum_{\emptyset\subset X\subseteq \{x+1,\dots,x'-1\}}\hspace{-2em}(-1)^{|X|} \left| \{g\in T[x,x', (A,B), f]:\: g(A) \text{ is sparse at all $m \in X$ } \} \right|.
\]
Denoting $X=\{x_1,\dots, x_\ell \}$, where $x_1<\dots<x_\ell$, we further rewrite this into

\[
\bigsum_{\ell=1}^{x'-x-2}(-1)^\ell\hspace{-1em}\bigsum_{x< x_1<\dots<x_\ell< x'}\hspace{-1em} |\{g\in T[x, x', (A,B), f]: g(A) \text{ is sparse at } x_1,\dots, x_\ell \}|.
\]

Now we claim that, since $Q\cap V\langle m\rangle$ is a separator of $G[Q]$, $|T[x, x', (A, B), f]|$ can be further rewritten to express it recursively as follows:
\begin{claim}
\[ \label{eq:inc_exc}
|T[x, x', (A, B), f]| = \bigsum_{\ell=1}^{x'-x-2} (-1)^\ell \sum_{(*)} \prod_{i=0}^{\ell} |T[x_i, x_{i+1}, (A_i, B_i), f_i]|,
\]
where we let $x_0=x$ and $x_{\ell+1}=x'$ for convenience and the sum $(*)$ goes over
\begin{itemize}
	\item integers $x< x_1<\dots<x_\ell< x'$,
	\item separations $(A_i, B_i)$ of $P$ of order $4\sqrt{pk}$ for each $i=0,\ldots,\ell$, such that
	\begin{itemize}
	    \item $\cup_{i=0}^\ell A_i=A$, and
	    \item $A_i\setminus B_i$ and $A_j\setminus B_j$ are disjoint for each $0\leq i< j \leq \ell$,
	\end{itemize}
	\item functions $f_i:A_i\cap B_i\rightarrow V\langle x_i\rangle\cup V\langle x_{i+1}\rangle$ for each $i=0,\dots, \ell$ such that $f,f_1,\ldots,f_{\ell}$ are pairwise compatible and $|f^{-1}_i(V\langle x_i\rangle)|, |f^{-1}_i(V\langle x_{i+1}\rangle )|\leq 2\sqrt{pk}$.
\end{itemize}
\end{claim}
\begin{claimproof}
To prove this claim, consider first a function $g\in T[x,x', (A,B), f]$ such that $g(A)$ is sparse at $x_1,\dots, x_{\ell}$. We describe how to find the  separations $(A_i, B_i)$ and functions $f_i$ that correspond to $g$. Let $A_i=g^{-1}(V\langle x_i, x_{i+1}\rangle)$, $B_i=(V(P)\setminus A_i)\cup g^{-1}(V\langle x_i\rangle)\cup g^{-1}(V\langle x_{i+1} \rangle)$. Note that, since $g(A)$ is sparse at $x_i$ and $x_{i+1}$, $(A_i, B_i)$ is a separation of order at most $4\sqrt{pk}$. It is easy to see that $\cup A_i=g^{-1}(V\langle x,x'\rangle)=A$. Also, note that $A_i\setminus B_i=g^{-1}(V\langle x_i, x_{i+1}\rangle)\setminus (g^{-1}(V\langle x_i\rangle)\cup g^{-1}(V\langle x_{i+1}\rangle))$, so for any $i\neq j$, $A_i\setminus B_i$ and $A_j\setminus B_j$ are disjoint. We define $f_i:A_i\cap B_i\rightarrow V\langle x_i\rangle \cup V\langle x_{i+1}\rangle$ as $f_i=g|_{A_i\cap B_i}$. By construction, $f, f_1,\dots, f_{\ell}$ are pairwise compatible. 

Conversely, given pairwise compatible functions $g_0,\dots, g_{\ell}$ such that  $g_i\in T[x_i, x_{i+1}, (A_i, B_i), f_i]$, we show how to construct a function $g\in T[x, x', (A,B), f]$. Since the $g_i$'s are compatible, we can define $g=g_0\cup \dots\cup g_{\ell}:A\rightarrow V\langle x,x'\rangle$. Since $f, g_1,g_{\ell}$ are pairwise compatible, $g$ extends $f$.
It is easy to see that this correspondence is one to one, which proves the claim.
\end{claimproof}
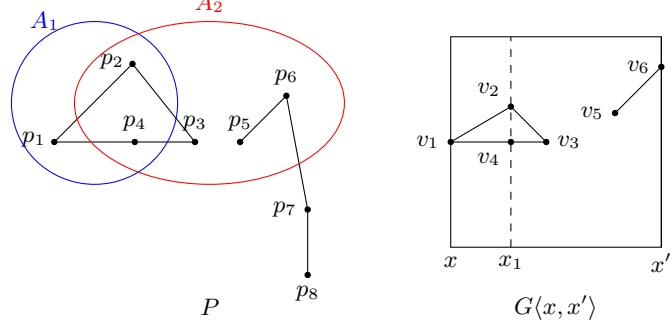
\begin{figure}
    \centering
    \vspace{-0.65cm}
    \begin{tikzpicture}
    	\def\u{0.4}
    	\tikzset{
    		vx/.style={draw, circle, inner sep=0pt, outer sep=0pt, fill, minimum size=2pt}
    	}
    	\draw (0,0) rectangle ++(7*\u, 7*\u);
    	\node[](x0) at (0,-0.5*\u){$x$};
    	\node[](x1) at (2*\u, -0.5*\u){$x_1$};
    	\node[](x2) at (7*\u, -0.5*\u){$x'$};
    	
    	\draw[dashed] (2*\u, 0) -- (2*\u, 7*\u);
    	\draw (7*\u, 0) -- (7*\u, 7*\u);
    	\node[vx, label=below left:$v_4$](v4) at (2*\u, 3.5*\u){};
    	\node[vx, label=left:$v_1$](v1) at (0, 3.5*\u){};
    	\node[vx, above=\u of v4, label=above left:$v_2$](v2){};
    	\node[vx, right=\u of v4, label=right:$v_3$](v3){};
    	\draw (v1) -- (v2) --(v3)--(v1);
    	
    	\node[vx, label=left:$v_6$](v6) at (7*\u, 6*\u){};
    	\node[vx, label=left:$v_5$, below left=2*\u of v6](v5){};
    	
    	\draw (v5) --(v6);
    	\node[vx, left=13*\u of v1, label={[left]:$p_1$}](p1){};
    	\node[vx, right=2.5*\u of p1, label=$p_4$](p4){};
    	\node[vx, right=4.5*\u of p1, label=$p_3$](p3){};
    	\node[vx, above right=3.5*\u of p1, label={[left]:$p_2$}](p2){};
    	\node[vx, right=6*\u of p1, label=$p_5$](p5){};
    	\node[vx, above right=2*\u of p5, label=$p_6$](p6){};
    	\node[vx, below right=3*\u of p5, label=left:$p_7$](p7){};
    	\node[vx, below=2*\u of p7, label=below:$p_8$](p8){};
    	\draw (p1) --(p2)--(p3)--(p1);
    	\draw (p5)--(p6)--(p7)--(p8);
    	
    	\node[draw=blue, ellipse, fit= (p1) (p2) (p4), inner sep=6pt, label={[blue, left=10pt]:$A_1$}](a1){};
    	\node[draw=red, ellipse, fit= (p2) (p3) (p4) (p5) (p6), inner sep=6pt, label={[red]:$A_2$}](a2){};
    	
    	\node[](p) at (-8*\u, -2*\u){$P$};
    	\node[](g) at (3.5*\u, -2*\u){$G\langle x,x' \rangle$};
    	
    \end{tikzpicture}
    \label{fig:alg}
    \caption{The function $g:\{p_1,\dots,p_6\}\rightarrow V\langle x,x'\rangle$ defined by $g(p_i)=v_i$, corresponds to functions $g_1:A_1\rightarrow V\langle x, x_1\rangle$ and $g_2:A_2\rightarrow V\langle x_1,x'\rangle$, where $g_1(p_i)=v_i$ and $g_2(p_i)=v_i$. }
    \vspace{-0.5cm}
\end{figure}
The next step is to rewrite the sum $(*)$ to match the form of Lemma~\ref{lem:inc_exc}. The only difference is that in~\eqref{eq:prod} the summation is over variables that are pairwise independent.

Formally, let us define a square matrix $M$ whose indices $M_{ind}$ are of the form $(x_i, (A_i, B_i), f_i)$, where $x_i\in\{x,\dots,x'\}$, $(A_i, B_i)\in \mathcal{S}_{2\sqrt{pk}}(P)$  and $f_i:A_i\cap B_i\rightarrow V\langle x_i\rangle\cup V\langle x_0\rangle$.
Let $I_i=f_i^{-1}(V\langle x_i\rangle)$, $I_j=f_j^{-1}(V\langle x_j\rangle)$.


If $x_j\geq x_i$, $f_i, f_j$ compatible and $|I_i|, |I_j|\leq \sqrt{pk}$ we define \\
$M[(x_i, (A_i, B_i), f_i), (x_j, (A_j, B_j), f_j)]$ as 
$$\mu((A_i, B_i), (A_j, B_j))\cdot T[x_i, x_j, ((A_j\setminus A_i)\cup I_i,  B_j\cup A_i), f_i|_{I_i}\cup f_j|_{I_j}],$$
and zero otherwise.

Intuitively, $M[(x_i, (A_i, B_i), f_i), (x_j, (A_j, B_j), f_j)]$ describes the number of ways to embed $P[A_j\setminus A_i]$ between lines $x_i$ and $x_j$, where $f_i$ and $f_j$ describe the behaviour of these embeddings on lines $x_i$ an $x_j$ respectively. We observe that we can group isomorphic separations together, i.e. that instead of indexing by every separation, we can index by their representatives and take into account the multiplicities, which are described by $\mu$.

Now we can rewrite the sum $(*)$ as
\[
\begin{aligned}
\sum_{(**)} &M[(x_\ell, (A\setminus A_{\ell-1}, B\setminus B_{\ell-1}), f_{\ell-1}), (x', (A, B), f)] \\
&\bigprod_{i=0}^{\ell-2} M[(x_i, (A_i, B_i), f_i), (x_{i+1}, (A_{i+1}, B_{i+1}, f_{i+1}))],
\end{aligned}
\]
where the sum $(**)$ goes over $(x_0, (A_0, B_0), f_0), \dots, (x_{\ell-1}, (A_{\ell-1}, B_{\ell-1}), f_{\ell-1})\in M_{ind}$. Now by Lemma \ref{lem:inc_exc}, we can compute the sum $(*)$ in time $\ell\cdot |M_{ind}|^2$. Let us bound the size of $M_{ind}$. Recall that $|\mathcal{S}_{2\sqrt{pk}}(P)|=\sigma_{2\sqrt{pk}}(P)$ and note that $V\langle x_i\rangle$ contains at most $O(k^2p)$ disks (since we can assume that $G$ can be drawn in a $(O(k)\times O(k))$-box with ply $p$ by Theorem \ref{thm:turing}). Thus we have 
\[
|M_{ind}|\leq k^2 \sigma_{2\sqrt{pk}}(P)\cdot (Ck^2p)^{\sqrt{pk}}
\]
for some constant $C$. Therefore, we can compute $|\sub(P, G)|$ in time
$$k^{O(\sqrt{pk})}\cdot p^{O(\sqrt{pk})}\cdot \sigma_{O(\sqrt{pk})}(P)^2,$$
which concludes the proof of Theorem~\ref{thm:main}.
\section{Proof of Theorem \ref{thm:isosep}: Bounding $\sigma_s$ in Unit Disk Graphs}
\label{sec:boundsigma}
In this section, we bound the parameter $\sigma_s$. We first present a bound for connected graphs.
\begin{theorem}
	For a connected unit disk graph $P$ of ply $p$ with $k$ vertices, and an integer $s$, we have $\sigma_{s}(P)\leq 2^{O(s \log k)}$.
\end{theorem}
\begin{proof}
	For a separations $(A,B)$ of $P$ of order $s$, there are $\binom{k}{s}\leq k^s\leq 2^{s \log k}$ possibilities for $S=A \cap B$. Since any vertex in a unit disk graph can be adjacent to at most $6$ pairwise non-adjacent neighbors, the number of connected components of $P - A\cap B$ is at most $6|A \cap B|$. Since $(A,B)$ is a separation, each connected component of $P-S$ is either fully contained in $A$, or disjoint from $A$. Therefore we only have at most $2^{6|A \cap B|}$ separations $(A,B) \in \mathcal{S}_s(P)$ with $A \cap B = S$. 
\end{proof}

\begin{theorem}
	For a unit disk graph $P$ of ply $p$ with $k$ vertices, we have that $\sigma_s(P)$ is at most $2^{O(s \log k +pk/\log k)}$.
\end{theorem}
\begin{proof}
	For a separations $(A,B)$ of $P$ of order $s$, there are $\binom{k}{s}$ possibilities for $S=A \cap B$.
	We define $t=\frac{\log k}{3p}$. We split all connected components $C$ of $P-S$ into 3 categories as follows:
	\begin{itemize}
		\item \emph{Adjacent Components}: there is an edge between $C$ and $S$;
		\item \emph{Small Components}: there is no edge between $C$ and $S$, and $C$ contains at most $t$ vertices;
		\item \emph{Large Components}: there is no edge between $C$ and $S$, and $C$ contains more than $t$ vertices.
	\end{itemize}
	Let $\mathcal{C}_{\mathrm{adj}}$ and $\mathcal{C}_{\mathrm{lar}}$ be the sets of all adjacent and large connected components of $G-S$ respectively. Let $\mathcal{C}_{\mathrm{sma}}$ be a maximal set of \emph{non-isomorphic} small connected components. That is, for each small component $C$ of $P-S$ there is exactly one component $C' \in \mathcal{C}_{\mathrm{sma}}$ such that $P[C]$ is isomorphic to $P[C']$.
	
	Note we can encode $(A,B)$ as a quadruple consisting of $S$, a subset of $\mathcal{C}_{\mathrm{lar}}$, $\mathcal{C}_{\mathrm{adj}}$, and for each $C \in \mathcal{C}_{\mathrm{sma}}$ a number that indicates how many isomorphic copies of $C$ are contained in $A$. The next step is to bound the number of such quadruples:
	
	\begin{equation}
		\label{eq:noseps}
		\binom{k}{s} 2^{|\mathcal{C}_{\mathrm{lar}}|+|\mathcal{C}_{\mathrm{adj}}|}k^{|\mathcal{C}_{\mathrm{sma}}|}.
	\end{equation}
	We have that $\mathcal{C}_{\mathrm{lar}}$ has at most $k/t$ elements since the components are disjoint and in total amount to at most $k$ vertices. We have that $|\mathcal{C}_{\mathrm{adj}}| \leq O(s)$ since each vertex of $S$ has at most $6$ neighbors that are pairwise non-adjacent (since $P$ is a unit disk graph). 
	
	By Lemma~\ref{lem:number_udg}, the number of unlabelled unit disk graphs with $t$ vertices and ply $p$ is at most $2^{(p+1)t}$ and therefore $|\mathcal{C}_{\mathrm{sma}}| \leq 2^{(p+1)t}$. Thus~\eqref{eq:noseps} is at most
	\[
	2^{s\log k+3pk/\log k + s+k^{1/3}\log k } = 2^{O(s \log k +pk/\log k)},
	\]
	as claimed.
\end{proof}

\section{Proof of Theorem \ref{thm:LB}: Lower bound}
\label{sec:LB}
In this section, we give a proof of Theorem \ref{thm:LB}, showing that under ETH there is no algorithm deciding whether $|\sub(P,G)| > 0$ ($|\ind(P,G)| > 0$ respectively) in time $2^{o(n /\log n)}$ even when the ply is two. We will use a reduction from the \textsc{String 3-Groups} problem similar to the one in~\cite{ICALP}.

\begin{definition}
The \textsc{String 3-Groups} problem is defined as follows. Given sets $A,B,C \subseteq \{0, 1\}^{6\lceil \log n\rceil+1} $ of size $n$, find $n$ triples $(a, b, c) \in A\times B \times C $ such that for all $i$, $a_i + b_i + c_i \leq 1$ and each element of $A,B,C$ occurs exactly once in a chosen triple.
\end{definition}

We call the elements of $A,B,C$ strings. It was shown in~\cite{ICALP} that, assuming the ETH, there is no algorithm that solves \textsc{String 3-Groups} in time $2^{o(n)}$. 
Before stating the formal proof of Theorem~\ref{thm:LB}, we give an outline of the main ideas.
Given an instance $(A,B,C)$ of \textsc{String 3-Groups problem}, we construct the corresponding host graph $G$ and pattern $P$ as follows. Firstly, we modify slightly the
strings in $A,B,C$ to facilitate the construction of $P$ and $G$. Let $m$ be the length of the (modified) strings. For each $a \in A$, the connected component in $G$ that corresponds to it consists of two paths $p_1 \dots p_m$ and $q_1 \dots q_m$, where $p_i$ and $q_i$
are connected by paths of length 3 if $a_i = 0$. For each $b \in B$, the connected component in $P$ that corresponds to it consists of a path $t_1 \dots t_m$, where there is a path of length two attached to $t_i$ if $b_i = 1$. The connected components corresponding to elements in $C$ are constructed in a similar way. Finally, we add gadgets (triangles and 4-cycles) to each connected component in $P$ and $G$ to ensure we cannot "flip" the components in $P$. For an example, see Figure \ref{fig:lowerb}.

\begin{figure}[h]
	\centering
	\begin{subfigure}{0.3\textwidth}
		\begin{tikzpicture}
			\def \u{0.55}
			\tikzset{
				Udisc/.style={circle, draw=black, minimum size=\u cm, inner sep=0, outer sep=0, font=\small}
			}
			\foreach \i in {1,...,7}
			{
				\node[Udisc](p\i) at (0, \i*\u){$p_\i$};
				\node[Udisc](q\i) at (4*\u, \i*\u){$q_\i$};		
			}
			\foreach \i in {1,5,7}
			{
				\node[Udisc](r\i) at (\u, \i*\u){$r_\i^{1}$};
				\node[Udisc](r\i') at (2*\u,\i*\u){$r_\i^{2}$};
				\node[Udisc](r\i'') at (3*\u, \i*\u){$r_\i^{3}$};
			}
			\node[Udisc, blue, left=0cm of p1](x1){$x_1$};
			\node[Udisc, blue, left=0cm of x1](x2){$x_2$};
			\node[Udisc, blue, below=0cm of x2](x3){$x_3$};
			\node[Udisc, blue, below=0cm of x1](x4){$x_4$};
			
			\node[Udisc, blue, right=0cm of q1](y1){$y_1$};
			\node[Udisc, blue, right=0cm of y1](y2){$y_2$};
			\node[Udisc, blue, below=0cm of y2](y3){$y_3$};
			\node[Udisc, blue, below=0cm of y1](y4){$y_4$};
			
			\node[Udisc, blue, right=0cm of q7](y1'){$y_1'$};
			\node[Udisc, blue, right=0cm of y1']{$y_2'$};
			\path (y1') ++ (60:\u) node[Udisc, blue](y3'){$y_3'$};
			
			\node[Udisc, blue, left=0cm of p7](x1'){$x_1'$};
			\node[Udisc, blue, left=0cm of x1'](x2'){$x_2'$};
			\path (x1') ++(120:\u) node[Udisc, blue](x3'){$x_3'$};
		\end{tikzpicture}
	\end{subfigure}
	\hfill
	\begin{subfigure}{0.5\textwidth}
		\begin{tikzpicture}
			\def\u{0.55}
			\tikzset{
				Udisc/.style={circle, draw=black, minimum size=\u cm, inner sep=0, outer sep=0}
			}	
			
			\foreach \i in {1,...,7}
			{
				\node[Udisc](s\i) at (0, \i*\u){$s_\i$};
				\node[Udisc](t\i) at (6*\u, \i*\u){$t_\i$};		
			}
			\foreach \i in {1,7}
			{
				\node[Udisc, right=0cm of s\i]($s_\i^1$) {$s_\i^1$};
				\node[Udisc, right=\u cm of s\i]($s_\i^2$) {$s_\i^2$};
			}
			\foreach \i in {5}
			{
				\node[Udisc, left=0cm of t\i](t\i'){$t_\i^1$};
				\node[Udisc, left=0cm of t\i']($t_\i''$){$t_\i^2$};
			}
			\node[Udisc, left=0 cm of s1, color=blue](z1){$z_1$};
			\node[Udisc, left=0cm of z1, color=blue](z2){$z_2$};
			\node[Udisc, below=0cm of z2, color=blue](z3){$z_3$};
			\node[Udisc, below=0cm of z1, color=blue](z4){$z_4$};
			
			\node[Udisc, right=0cm of t1, color=blue](w1){$w_1$};
			\node[Udisc, right=0cm of w1, color=blue](w2){$w_2$};
			\node[Udisc, below=0cm of w2, color=blue](w3){$w_3$};
			\node[Udisc, below=0cm of w1, color=blue](w4){$w_4$};
			
			\node[Udisc, left=0cm of s7, blue](z1'){$z_1'$};
			\node[Udisc, left=0cm of z1', blue](z2'){$z_2'$};
			\path (z1') ++(120:\u) node[Udisc, blue](z3'){$z_3'$};
			
			\node[Udisc, blue, right=0cm of t7](w1'){$w_1'$};
			\node[Udisc, blue, right=0cm of w1'](w2'){$w_2'$};
			\path (w1') ++(60:\u) node[Udisc, blue](w3'){$w_3'$};	
		\end{tikzpicture}
	\end{subfigure}
	\caption{Connected components corresponding to $a=011 101 0\in A$ (left), $b=100 000 1\in B$ (middle), $c=000 010 0\in C$ (right). }
	\label{fig:lowerb}
\end{figure}
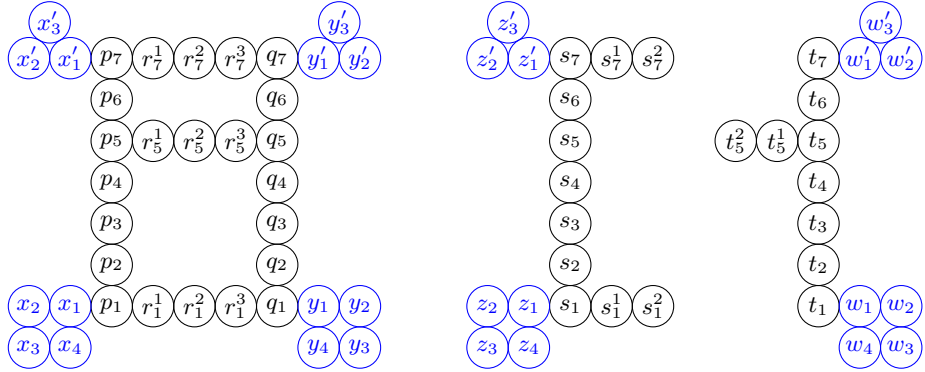
We present only the proof for the $\ind$ case, as the proof for $\sub$ is analogous.
\begin{theorem}
	Assuming ETH, there is no algorithm to determine if $\ind(P, G)$ is nonempty for given unit disk graphs $P$ and $G$ in $2^{o(n/\log n)}$ time (where $n=|V(G)|$), even when $P$ and $G$ have a given embedding of ply 2.
\end{theorem}
\begin{proof}
	Given an instance $A_0, B_0, C_0$ of the \textsc{String 3-Groups} problem, we show how to construct an equivalent instance of subgraph isomorphism with host graph $G$ and pattern $P$. We will first construct another instance $A,B,C$ of the \textsc{String 3-Groups} problem by modifying the given strings as follows: 
	$$A=\{\tilde{a}:a\in A_0\}, \: B=\{\tilde{b}:b\in B_0\}, \: C=\{\tilde{c}: c\in C_0\},$$
	where 
	\begin{itemize}
		\item $\tilde{a}$ is constructed from $a$ by inserting a 1 between every two consecutive characters and 10 at the end;
		\item $\tilde{b}$ is constructed from $b$ by inserting a 0 between every two consecutive characters and 01 at the end;
		\item $\tilde{c}$ is constructed from $c$ by inserting a 0 between every two consecutive characters and 00 at the end.
	\end{itemize}
	For example, if $a=010\in A_0$, $b=100\in B_0$, $c=001\in C_0$, then $\tilde{a}=0\underline{1}1\underline{1} 0\underline{10}$, $\tilde{b}=1\underline{0}0\underline{0} 0\underline{01}$, $\tilde{c}=0\underline{0}0\underline{0} 1\underline{00}$.
	
	It is easy to see that for all $(a,b,c)\in A_0\times B_0\times C_0$ we have $a+b+c\leq \mathbf{1}$ if and only if $\tilde{a}+\tilde{b}+\tilde{c}\leq \mathbf{1}$. Therefore, it is enough to construct an instance of subgraph isomorphism equivalent to $A,B,C$. Let $m=12\lceil\log n \rceil+3$ denote the length of strings in $A,B,C$. 
	
	The graph $G$ will consist of $n$ connected components, representing the elements of $A$. The graph $P$ will consist of $2n$ connected components, representing elements of $B$ and $C$. For each string $a\in A$, add to $G$ a connected component that is constructed as follows: 
	Construct two paths, $p_1\dots p_m$ and $q_1\dots q_m$. Whenever $a_i=0$, add a path $p_ir_i^{1}r_i^{2}r_i^{3}q_i$. Construct four cycles: $x_1x_2x_3x_4$, $y_1y_2y_3y_4$, $x_1'x_2'x_3'$ and $y_1'y_2'y_3'$ and add edges $x_1p_1$, $x_1'p_m$, $y_1q_1$, $y_1'q_m$. 
	
	Note that since $a_2=0$, so the above paths are connected. By construction, there is no $i$ such that $a_i=a_{i+1}=0$ (i.e. there is no $i$ such that both $r_i^{1}r_i^2r_1^3$ and $r_{i+1}^1r_{i+1}^2r_{i+1}^3$ exist), so we can embed $G$ on the plane as shown in Figure \ref{fig:lowerb}. 
	
	Let us now construct the graph $P$. For each string $b\in B$, add to $P$ a connected component that is constructed as follows. Construct a path $s_1\dots s_m$. Whenever $b_i=1$, add a path $s_is_i^1s_i^2$. Construct cycles $z_1z_2z_3z_4$ and $z_1'z_2'z_3'$ and add edges $z_1s_1$ and $s_mz_1'$. 
	
	Similarly, for each $c\in C$, construct a path $t_1\dots t_m$. Whenever $c_i=1$, add a path $t_it_i^1t_i^2$. Construct cycles $w_1w_2w_3w_4$ and $w_1'w_2'w_3'$ and add edges $w_1t_1$ and $t_mw_1'$. 
	
	Again by construction, there is no $i$ such that $b_i=b_{i+1}=1$ or $c_i=c_{i+1}=1$, so we can embed $P$ in the plane as in Figure \ref{fig:lowerb}.  
	
	Given a solution to this \textsc{String 3-Groups} instance, we can embed $P$ into $G$ as follows. If $(a,b,c)\in A\times B\times C$ is a triple in the solution, we map the components corresponding to $b$ and $c$ to the component corresponding to $a$, by mapping:
	\begin{itemize}
		\item $s_i$ to $q_i$ and $t_i$ to $p_i$ for $i=1,\dots,m$
		\item $z_i$ to $y_i$, $w_i$ to $x_i$ for $i\in\{1,2,3,4\}$ and $z_i'$ to $y_i'$, $w_i'$ to $x_i'$ for $i\in\{1,2,3\}$
		\item $s_i^1, s_i^2$ to $r_i^3, r_i^2$, and $t_i^1, t_i^2$ to $r_i^1, r_i^2$ for all $s_i^1, s_i^2, t_i^1,t_i^2\in V(P)$
	\end{itemize}
	This map is well defined: indeed, $a_i+b_i+c_i\leq 1$ for each $i$, so whenever $s_i^1, s_i^2\in V(P)$ or $t_i^1,t_i^2\in V(P)$, we have $r_i^1,r_i^2,r_i^3\in V(G)$. Also, at most one of $b_i$ and $c_i$ is equal to one, so we either have $s_i^1,s_i^2\in V(P)$ or $t_i^1,t_i^2\in V(P)$ (or possibly neither), i.e. the map is injective. It is easy to check that the above map maps $P$ to an induced subgraph of $G$.
	
	Conversely, given an injective homomorphism $f:P\rightarrow G$, we can construct a solution to the \textsc{String 3-Groups} instance as follows. For each $a\in A$, we describe how to find $b\in B$ and $c\in C$ such that $a_i+b_i+c_i\leq 1$ for all $i$.
	
	By a counting argument it is impossible to map more than two connected components of $P$ to the same connected component of $G$. Since the number of connected components in $G$ and $P$ is $n$ and $2n$ respectively, each connected component in $G$ is therefore the image of exactly two connected components in $P$. 
	Let $a\in A$, and let $G_a$ be the corresponding connected component in $G$. We denote the paths of length $m$ in $G_a$ by $p_1\dots p_m$ and $q_1\dots q_m$, the paths between $p_i$ and $q_i$ by $r_i^1r_i^2r_i^3$, the cycles by $x_1x_2x_3x_4$, $y_1y_2y_3y_4$, $x_1'x_2'x_3'$, $y_1'y_2'y_3'$ as in the above construction. 
	
	Let $P_1$ and $P_2$ the connected components of $P$ that map to $G_a$. Note that $G_a$ has exactly two 4-cycles and two triangles, while $P_1$ and $P_2$ have one 4-cycle and one triangle each. Therefore, without loss of generality we can assume that the 4-cycle in $P_1$ is mapped to $y_1y_2y_3y_4$, and the 4-cycle in $P_2$ to $x_1x_2x_3x_4$. Denote the paths of length $m$ in $P_1$ ($P_2$ respectively) by $\bar{s}_1\dots \bar{s}_m$ ($\bar{t}_1\dots \bar{t}_m$ respectively). Denote the paths of length 3 by $\bar{s}_i\bar{s}_i^1\bar{s}_i^2$ ($\bar{t}_i\bar{t}_i^1\bar{t}_i^2$ respectively).
	
	\begin{claim}
		We have $f(\bar{s}_i)=q_i$ for all $i\in [m]$. 
	\end{claim}
	\begin{claimproof}
		We know that $f(\bar{s}_1)=q_1$ (since it is the only vertex that is adjacent to the image of the 4-cycle in $P_1$). Similarly, $f(\bar{t}_1)=p_1$, and $\bar{s}_m$ and $\bar{t}_m$ are mapped to  $q_m$ and $ p_m$ (not necessarily in that order).  
		Suppose that $f(\bar{s}_m)=p_m$ and $f(\bar{t}_m)=q_m$. Every path of length $m$ between $q_1$ and $p_m$ contains a "vertical line", i.e. a subpath $\pi=q_ir^1_ir^2_ir^3_ip_i$ for some $i$. Note that $G_a-\pi$ has two connected components, one containing $p_1$ and the other containing $q_m$. This leads to a contradiction, since $f(\bar{t}_1)=p_1$ and $f(\bar{t}_m)=q_m$ need to be connected in $f(P_1)\subseteq G_a-\pi$. Therefore, we have $f(\bar{s}_m)=q_m$ and $f(\bar{t}_m)=p_m$. Using a similar argument, we can conclude that $f(\bar{s}_i)=q_i$ for all $i\in [m]$, which proves the claim.
	\end{claimproof}
	
	Therefore, we have $f(\bar{s}^j_i)=r^{j}_i$, $f(\bar{t}^j_i)=r^j_i$ for all $i$ and $j=1,2,3$. Let us now look at the strings in $B\cup C$ that correspond to $P_1$ and $P_2$. Suppose both $P_1$ and $P_2$ correspond to strings in $B$. Since all strings in $B$ end with a 1, we have both $\bar{s}_m^2\in V(P_1)$ and $\bar{s}_m^2\in V(P_2)$, both of which are mapped to $r_i^2$, which leads to a contradiction. Therefore, at most one of $P_1$ and $P_2$ corresponds to a string in $B$. By a counting argument, we conclude that exactly one of $P_1$ and $P_2$ corresponds to a string  in $B$, while the other one corresponds to a string in $C$. Let $P_1$ correspond to $b\in B$ and $P_2$ to $c\in C$.   
	
	It remains to show that $a+b+c\leq \mathbf{1}$. By construction, whenever $a_i=1$, the vertex $r_i^2$ does not exist in $G_a$, so neither $\bar{s}_i^2$ nor $\bar{t}_i^2$ exist in $P$. Therefore, $b_i=c_i=0$. Since $f$ is injective, whenever $a_i=0$, at most one of $\bar{s}_i^2$ and $\bar{t}_i^2$ exists in $P_1\cup P_2$, i.e. at most one of $b_i$ and $c_i$ is equal to one. Hence this way we can construct the $n$ triples $(a,b,c)\in A\times B\times C$ such that $a+b+c\leq \mathbf{1}$. 
\end{proof}
\section{Conclusion}
\label{sec:concUDG}

We gave (mostly) sub-exponential parameterized time algorithms for computing $|\sub(P,G)|$ and $|\ind(P,G)|$ for unit disk graphs $P$ and $G$.
Since the fine-grained parameterized complexity of the subgraph isomorphism problem was only recently understood for planar graphs, we believe our continuation of the study for unit disk graphs is very natural, and we hope it inspires further general results.

While our algorithms are tight in many regimes, they are not tight in all regimes. In particular, the (sub)-exponential dependence of the runtime in the ply is not always necessary: We believe the answer to this question may be quite complicated: For detecting some patterns, such as paths, $2^{O(\sqrt{k})}n^{O(1)}$ time algorithms are known~\cite{DBLP:journals/jocg/ZehaviFLP021}, but it seems hard to extend it to the counting problem (and to all patterns with few non-isomorphic separations of small order).

For counting induced occurrences with bounded clique size our method can be easily adjusted to get a better dependence in the ply: namely, our method can be used to a get a $(kp)^{O(\sqrt{k})}$ time algorithm for counting independent sets of size $k$ in unit disk graphs of ply $p$ (which is optimal under the ETH by~\cite{DBLP:conf/focs/Marx07a}); is there such an improved independence on the ply for each pattern $P$?

It would be interesting to study the complexity of computing $|\sub(P,G)|$ and $|\ind(P,G)|$ for various pattern classes and various other geometric intersection graphs as well. Our results can be adapted to disk graphs where the ratio of the largest and smallest radius is constant (using a slight modification of Lemma~\ref{lem:number_udg}). 

A possible direction for further research would be to determine for which patterns can one compute the above values on bounded ply disk graphs?  Recent work~\cite{DBLP:conf/soda/LokshtanovPSXZ22} shows some problems admit algorithms running in sub-exponential time parameterized time.

Another direction would be to study the subgraph isomorphism problem for more general graphs, e.g. intersection graphs of fat objects. Most of our proofs rely only on a small subset of properties of bounded ply unit disk graphs, e.g. that a line segment of fixed length can only intersect a small number of disks. These properties are not unique to unit disk graphs of bounded ply, but also hold for a larger class of graphs, namely intersection graphs of fat objects. 

An even more general question would be to study subgraph isomorphism in higher dimensions. In \cite{chan2023finding}, it was shown that given an intersection graph $G$ of $n$ fat objects in dimension $d$ and a $k$-vertex graph $X_k$, one can determine whether $X_k$ is a subgraph of $G$ in time $O(n\log n )$ for constant $k$. Since the proof in \cite{chan2023finding} uses several techniques which are similar to those in this paper, it is natural to ask whether these can be used to obtain a parameterized algorithm for this more general graph class.

\bibliography{udg}
\end{document}

%% file: Macros.tex
\usepackage{enumitem}
\usepackage{wrapfig}
\usepackage{subcaption}
\usepackage[T1]{fontenc}
\usepackage[utf8]{inputenc}
\usepackage{todonotes}
\usepackage{comment}
\usepackage{soul}

\usepackage{cite}

\usepackage{amsmath, amssymb, graphics, graphicx}
\usepackage{color}
\usepackage{cases}
\usepackage{xparse}
\usepackage{xargs}
\usepackage{appendix}
\usepackage{xcolor}

\usepackage{relsize}
\newcommand{\bigsum}{\mathlarger{\sum}}
\newcommand{\bigprod}{\mathlarger{\prod}}
\newcommand{\ind}{\mathrm{ind}}
\newcommand{\sub}{\mathrm{sub}}

\usepackage{amsthm}
\newtheorem{observation}{Observation}
\theoremstyle{remark}
\newtheorem*{claimproof}{Proof of Claim}

\usepackage{tikz, standalone}
\usetikzlibrary{calc}
\usetikzlibrary{positioning, shapes.geometric, shapes, fit}

\renewcommand{\mod}{\ensuremath{\mathrm{mod}}}

\newcommand{\im}{\ensuremath{\mathbf{Im}}}

\makeatletter
\newtheorem*{rep@theorem}{\rep@title}
\newcommand{\newreptheorem}[2]{%
\newenvironment{rep#1}[1]{%
 \def\rep@title{#2 \ref{##1}}%
 \begin{rep@theorem}}%
 {\end{rep@theorem}}}
\makeatother